\documentclass[11pt,onecolumn,draftclsnofoot]{IEEEtran}
\usepackage{graphicx} \graphicspath{{./figs/}}
\usepackage{amsmath,amssymb} \interdisplaylinepenalty=2500
\usepackage{cite,bm}
\usepackage{verbatim}
\usepackage{subfig}

\newtheorem{definition}{Definition}
\newtheorem{theorem}{Theorem}

\newtheorem{myexample}{Example}
\newenvironment{example}{\myexample}{\qed\endmyexample}

\def\qed{\endIEEEproof}

\DeclareMathOperator{\rank}{\sf rank\hspace{0.1em}}

\DeclareMathOperator{\supp}{\sf supp}

\newcommand{\Fq}{\mathbb{F}_q}

\newcommand{\mat}[1]{\begin{bmatrix} #1 \end{bmatrix}}

\newcommand{\calC}{\mathcal{C}}

\title{Sparse Network Coding with Overlapping Classes}

\author{
\IEEEauthorblockN{Danilo Silva\hspace{0.05em}${}^1$, Weifei Zeng%\hspace{0.05em}${}^2$%
, and Frank R. Kschischang\\}
\IEEEauthorblockA{Department of Electrical and Computer Engineering, University of Toronto \\
Toronto, Ontario M5S 3G4, Canada\\
{\tt danilo@comm.utoronto.ca, weifei.zeng@utoronto.ca, frank@comm.utoronto.ca}
}
\thanks{${}^1$ Supported by CAPES Foundation, Brazil.}
\vspace{-3ex}
}

\begin{document}
\maketitle
\thispagestyle{empty}

\begin{abstract}

This paper presents a novel approach to network coding for distribution of large files. Instead of the usual approach of splitting packets into disjoint classes (also known as generations) we propose the use of \emph{overlapping} classes. The overlapping allows the decoder to alternate between Gaussian elimination and back substitution, simultaneously boosting the performance and reducing the decoding complexity.
Our approach can be seen as a combination of fountain coding and network coding.
Simulation results are presented that demonstrate the promise of our approach.

\end{abstract}

\section{Introduction}

Network coding \cite{Ahlswede++2000,Li++2003,Koetter.Medard2003,Ho++2006,Chou++2003} is a promising approach to data dissemination over networks. In this past decade, several works have attempted to establish the potential of this simple and yet seemingly revolutionary idea in a variety of applications \cite{Gkantsidis.Rodriguez2005:Avalanche,Gkantsidis++2006:LiveP2PSystem,Katti++2008:XORs,Fragouli++2008:EfficientBroadcasting,Wang.Li2007:R2,Dimakis++2007:INFOCOM}. While the success of network coding for streaming media and wireless applications has been encouraging, it is still unclear whether this approach is beneficial for peer-to-peer file dissemination \cite{Chiu++2006:CanNetworkCodingHelpP2P}. The present paper is an initial attempt to fill this gap.

%While network coding has been successfully applied to streaming media scenarios, several challenges arise when one attempts to use network coding for peer-to-peer file distribution.

One major issue is decoding complexity.
In the file-downloading scenario, a large file of $km \log_2 q$ bits is to be distributed among cooperating peers in a network. The file is partitioned into $k$ packets, each consisting of $m$ symbols over a finite field $\Fq$. If random linear network coding \cite{Ho++2006} is used to distribute the file, then each receiver has to solve a linear system with $k$ equations in order to decode the file. This requires $O(k^3 + k^2m)$ operations in $\Fq$, which may be prohibitively expensive in practice.

To reduce the decoding complexity, Chou \emph{et al.} \cite{Chou++2003} proposed to group packets into disjoint \emph{generations}, each containing $d$ packets, and apply network coding only within each generation. The complexity issue is solved if $d$ is small, but another problem is created: that of efficiently routing $L = k/d$ generations throughout the network.

Note that simply choosing a small $k$ and compensating the file size by using a large packet length $m$, as done in \cite{Gkantsidis++2006:LiveP2PSystem}, may not be a satisfactory solution. Transmitting such large packets (of, say, 1--4 MBytes \cite{Gkantsidis++2006:LiveP2PSystem}) over a \emph{dynamic} peer-to-peer network---where peers may interrupt transmissions or leave the network at any time---is a highly nontrivial problem. Since each \emph{coded} packet is essentially \emph{unique}, interrupted transmissions are useless to a receiving peer, potentially causing a severe waste of bandwidth. Thus, we find it more realistic to assume that $m$ is small and $L$ is large.

%At one hand, one could choose a relatively small $L$, and accommodate the file by using a large packet length $m$. Unfortunately, this creates a highly nontrivial problem: transmitting large, coded packets (of, say, ??? bytes []) in a dynamic peer-to-peer network (where peers may interrupt transmissions or leave the network at any time). Since each coded packet is essentially unique, interrupted transmissions are useless to the receiving peer, potentially causing a severe waste of bandwidth.

%However, transmitting such rather large packets (of, say, ??? bytes) between peers is a nontrivial problem, which is aggravated by the fact that, in a dynamic network, peers can interrupt transmissions or leave the network at any time. Since each coded packet is essentially unique, interrupted transmissions are useless to the receiving peer, potentially causing a severe waste of bandwidth.

%If, instead, one wishes to avoid this problem by using a small $m$, then $L$ must be large, which in turn creates a scheduling problem among generations.

%On the other hand, if this problem is avoided by using a small $m$ and a large $L$, then a scheduling problem among generations must be solved.
Probably the most successful approach so far to routing pieces of a file through a peer-to-peer network is the BitTorrent protocol \cite{BitTorrent}. The drawback of this and similar protocols is that a large number of control messages must be exchanged between peers, mainly to resolve the problems of \emph{rare blocks} and \emph{block reconciliation} \cite{Bharambe++2006:BitTorrentPerformance,Xu++2008:Swifter}. Thus, the protocol overhead is substantial, and a significant amount of research has been devoted to trying to alleviate this problem \cite{Bharambe++2006:BitTorrentPerformance}.

The solution proposed by Maymounkov \emph{et al.} \cite{Maymounkov++2006:ChunkedCodes}, in the context of generation-based network coding, completely eliminates any protocol overhead: peers randomly choose the generation from which to transmit each packet. This scheme is called \emph{chunked coding}. Intuitively, the scheme replaces %Intuitively, this scheme---called chunked coding---replaces
protocol overhead with transmission overhead.
While the scheme is shown to have a good performance asymptotically, the performance quickly deteriorates for practical values of $d$.

%Indeed, while the performance is shown to be good asymptotically, it quickly deteriorates for practical values of $d$.

%peers randomly choose the generations from which to transmit packets, irrespectively of the state of the receiving peers.

%each peer simply transmit generations chosen uniformly at random.
%In order to eliminate the protocol overhead, Mayam... proposed
%Their scheme, called chunked coding, is shown to have a good performance asymptotically. However, for practical values of $d$, the scheme has a large transmission overhead.

%The issue of decoding complexity has been dealt with in the fountain coding literature by using belief propagation decoding with optimized degree distributions.
A related line of work is fountain coding \cite{Shokrollahi2006}. By using optimized degree distributions, fountain codes such as LT or raptor codes can achieve a relatively small overhead with a low-complexity back-substitution decoder \cite{Shokrollahi2006}.
These schemes, however, are not compatible with network coding.
%, as packets must travel intact throughout the network---otherwise the degree distributions would be completely destroyed.
To maintain the designed degree distributions, packets must travel intact throughout the network---otherwise, the decoder would fail miserably.

% In order to use fountain coding in a network, we would have to use routing. Is it possible to combine fountain coding and network coding? Chunked codes offered this promise, but, for a fixed decoding complexity, the scheme incurs a large overhead (i.e., hardly achieves capacity). Moreover, a main difference between chunked coding and fountain coding is that %the former never makes use of back-substitution. the latter makes use of back-substitution to simplify the decoding, while the former does not.

This paper investigates the following question: is it possible to use a true network coding approach and yet enjoy a low-complexity fountain-like decoder? The approach proposed here answers this question affirmatively, and can be seen as a combination of fountain coding and network coding. Our idea is to follow the approach of chunked coding, but instead use a larger number of \emph{overlapping} generations (here called \emph{classes}). Overlapping generations allow packets from decoded generations to be back-substituted into still undecoded generations, in the same spirit of a fountain decoder. This not only boosts the performance but also reduces the decoding complexity of the scheme.

The remainder of the paper is organized as follows. In Section~\ref{sec:preliminaries}, we review previous work on network coding in a way that simplifies the description of our codes and emphasizes the existing connections. Section~\ref{sec:sparse-overlapping} presents our approach, including the description of the decoder and bounds on the decoding complexity. In Section~\ref{sec:examples}, we present some code constructions, whose performance is evaluated in Section~\ref{sec:performance} and compared with that of chunked codes. Finally, Section~\ref{sec:conclusion} presents some concluding remarks.

%For a full version of this paper, see \cite{Silva++2009:Overlapping-FullVersion}.

\section{Preliminaries}
\label{sec:preliminaries}

\subsection{Random Linear Network Coding}
\label{ssec:random-linear-network-coding}

Consider a communication network represented by a directed multigraph (cyclic or acyclic). The network is used to transport $k$ \emph{data} (or \emph{uncoded}) \emph{packets} $u_1,\ldots,u_k$ from a single source node to multiple destination nodes. Packets are regarded as vectors of length $m$ over a finite field $\Fq$. Each edge in the network is assumed to transport a \emph{single} packet, free of errors. To describe the operation of the network,
we associate with each edge $e$ a tuple $(P_e,t_e^-,t_e^+)$; if $e$ is an edge from a node $v^-$ to a node $v^+$, then this tuple indicates that packet $P_e$ was transmitted by $v^-$ at time $t_e^-$ and was received by $v^+$ at time $t_e^+$. We may also say that $P_e$ is an outgoing packet of $v^-$ and an incoming packet of $v^+$. For consistency, we assume that the data packets $u_i$ were \emph{received} by the source node at time $-\infty$.

The computation performed at the nodes must satisfy the \emph{law of (causal) information flow}: a packet transmitted by a node must be computed as a function of packets \emph{previously} received by that node. A (causal) schedule for a network is a specification of all the time values $t_e^-$, $t_e^+$ satisfying the constraint $t_e^- < t_e^+$.

Given a network and a schedule, a \emph{network code} is the specification of all functions computed at all nodes. In a \emph{linear network code} \cite{Li++2003,Koetter.Medard2003}, all such functions are constrained to be $\Fq$-linear combinations. This implies that any packet $P_e$ transmitted over the network can be expressed as a unique linear combination of data packets, say, $P_e = \sum_{i=1}^k g_{e,i} u_i$. The coefficient vector $g_e = (g_{e,1},\ldots,g_{e,k}) \in \Fq^k$ is called the \emph{(global) coding vector} of $P_e$.%, and is denoted $g(z) = g$.

Let $x_1,\ldots,x_N$ denote the outgoing packets of the source node, and let $y_1,\ldots,y_n$ denote the incoming packets of some destination node. Due to the linearity of the network code, these packets can be related by
\begin{equation}\label{eq:matrix-model}
  Y = AX = ABU
\end{equation}
where $U \in \Fq^{k \times m}$, $X \in \Fq^{N \times m}$ and $Y \in \Fq^{n \times m}$ are matrices whose rows are the packets $u_i$, $x_i$ and $y_i$, respectively, and $A \in \Fq^{n \times N}$ and $B \in \Fq^{N \times k}$. The matrix $AB$ is called the \emph{transfer matrix} of the network.

Note that successful decoding is possible if and only if $\rank AB = k$. In this case, the network code is said to be \emph{feasible}. Let $k^*$ denote the maximum rank of $A$ among all choices of the network code. Clearly, a feasible network code exists only if $k \leq k^*$, a condition we assume hereafter.

%A practical yet effective way of choosing the network code is a random choice.
In \emph{random linear network coding} \cite{Ho++2006}, nodes choose the coefficients of the linear combinations uniformly at random from $\Fq$ and independently from each other. As shown in \cite{Ho++2006}, a random network code is feasible with high probability if the field size $q$ is sufficiently large.

In order for the destination node to be able to undo the multiplication by $AB$ (which is unknown a priori) and recover $U$, the usual approach is to record the transfer matrix as part of the matrix $Y$ through the use of packet headers; more precisely, the left portion of $U$ is assumed to be a $k \times k$ identity matrix. Note that this leaves space for only $m' = m - k$ data symbols in each data packet, i.e., the effective throughput is scaled by $\frac{m-k}{m}$. In practice, one must choose $m' \gg k$.

Decoding corresponds to applying Gauss-Jordan elimination on $Y$ to convert it to reduced row echelon form. Note that only $k$ linearly independent rows of $Y$ are effectively needed. Performing Gauss-Jordan elimination on a $k \times (k+m')$ matrix requires $k^2 m' + \frac{1}{2} k^2(k-1)$ multiplications and a similar number of additions\footnote{Note that asymptotically fast methods are only useful for very large parameters (much larger than those consider in this paper).}. We will ignore the number of additions since the time to perform an addition is usually negligible compared to the time to perform a multiplication. We also ignore the second term in the operation count since, as discussed above, $m' \gg k$ in any realistic scheme. Thus we may say that the decoding complexity of random linear network coding is $k$ operations per data symbol.

Due to the fact that the transfer matrix $AB$ is dense, this scheme is also called \emph{dense network coding}.

\subsection{Sparse Network Coding with Disjoint Classes}
\label{ssec:sparse-disjoint}

For large $k$, dense network coding is computationally too expensive in practice. A way to alleviate this problem is to ensure that $AB$ has a sparse structure. The main difficulty is that this constraint must be not only imposed at the source node, but also coordinated among all the internal nodes---which must still be able to perform network coding.

An approach proposed in \cite{Chou++2003} is to divide packets into disjoint \emph{classes} (or \emph{generations} \cite{Chou++2003}, \emph{groups} \cite{Gkantsidis++2006:LiveP2PSystem}, \emph{segments} \cite{Wang.Li2007:R2}, \emph{chunks} \cite{Maymounkov++2006:ChunkedCodes}). Suppose that $k = L d$. For $i = 1,\ldots,k$, let us say that a packet $u_i$ belongs to class $\ell$ if $i \in \{(\ell-1)d+1,\ldots,\ell d\}$. Now, the rule that is enforced at each network node is that \emph{only packets of the same class are allowed to be combined}, producing a new packet of the same class. Under this constraint, expression (\ref{eq:matrix-model}) can be rewritten
%(up to reordering of rows)
as
%\begin{equation}\nonumber
%  \mat{Y^{(1)} \\ \vdots \\ Y^{(L)}} = \mat{A^{(1)} & 0 & 0 \\ 0 & \ddots & 0 \\ 0 & 0 & A^{(L)}} \mat{B^{(1)} & 0 & 0 \\ 0 & \ddots & 0 \\ 0 & 0 & B^{(L)}} \mat{U^{(1)} \\ \vdots \\ U^{(L)}}
%\end{equation}
\begin{equation}\nonumber
  Y^{(\ell)} = A^{(\ell)}B^{(\ell)}U^{(\ell)}, \quad \ell=1,\ldots,L
\end{equation}
where
\begin{equation}\nonumber
  U^{(\ell)} = \mat{u_{(\ell-1)d+1} \\ \vdots \\ u_{\ell d}} \quad \text{and} \quad Y^{(\ell)} = \mat{y_{1}^{(\ell)} \\ \vdots \\ y_{n_\ell}^{(\ell)}},
\end{equation}
and where $y_j^{(\ell)}$, $j=1,\ldots,n_\ell$, $\ell=1,\ldots,L$, are the received packets.
Note that this is essentially splitting the network into $L$ parallel smaller networks. Due to the block-diagonal structure of $AB$, decoding can now be performed in $\frac{1}{km} L d^2 m = d$ operations per symbol, which may be a dramatic improvement if $L$ is large.

Increasing $L$ also reduces the overhead in transmitting packet headers. Rather than $k$ symbols per packet, the overhead is now only $\lceil \log_q L \rceil + d$ symbols per packet, corresponding to a class index plus a coding vector.

The performance of this scheme, however, reduces as $L$ increases.
This is mainly due to the following reasons. First, separating flows into disjoint classes reduces the diversity of source-destination paths, which may reduce the min-cut of the network (and therefore $k^*$). Second, the fact that fewer packets are combined together within each class may increase the probability of linear dependency among received packets. Third, differently from the $L=1$ case, nodes have to \emph{choose} the class from which to produce a new packet at each transmission opportunity. This implies that the induced network topology is chosen by the nodes on-the-fly, and poor choices may lead to a poor overall system. Fourth, the decoding condition is ``$L$~times more constrained:'' decoding is successful if and only if $\rank A^{(\ell)}B^{(\ell)} = d$ for $\ell=1,\ldots,L$.

The first and second problems are mitigated if $k$ and $q$, respectively, are sufficiently large. For the third problem, different strategies have been proposed, most of which require exchange of control messages. We will focus here on the strategy proposed in \cite{Maymounkov++2006:ChunkedCodes}, which eliminates any need for feedback:
%
%. One approach in [] is to transmit one class after another, but this is inefficient in the file downloading scenario: a node that lacks a single packet in a class has to wait until a new transmission cycle beings (yielding a very large overhead). If nodes can communicate through a control channel, then specific classes may be transmitted at the request of downstream nodes []; however, control messages increase the communication overhead, which may again lead to inefficiency if $L$ is large.
%
%The strategy proposed in [] eliminates any need for feedback:
nodes simply choose classes uniformly at random among previously received classes. This scheme is referred to as \emph{chunked coding}. The drawback of this approach is that it exacerbates the fourth problem. A node may unnecessarily receive packets from a class that has already been fully decoded, while other classes are still incomplete; this in turn requires $n$ to be much larger than $k$. The results in \cite{Maymounkov++2006:ChunkedCodes} show that the overhead $(n-k)/k$ can be made comparatively small by choosing $d = \ln^2 k$ and letting $k$ be sufficiently large. In practice, however, such a large $d$ defeats the purpose of sparse network coding, since the decoding complexity becomes prohibitively large.

The bottom line for this approach of dividing packets into disjoint classes is that it simply \emph{postpones} the scheduling problem: now \emph{classes} have to be routed, rather than individual packets. Thus, if $L$ is large, the same criticisms for any routing (non-network-coding) approach also apply here.

\section{Sparse Network Coding with Overlapping Classes}
\label{sec:sparse-overlapping}

%network-coding-compatible fountain codes

In this section we present a novel scheme that attempts to overcome the drawbacks of chunked coding. From one perspective, the scheme can be seen as a fountain code that is fully compatible with network coding.

In the following, the term \emph{class} refers to a non-empty subset of $\{1,\ldots,k\}$. A \emph{class-based scheme} for network coding is specified by a set of classes, $\calC = \{C_1,\ldots,C_L\}$, and a probability distribution on classes, $\{p_1,\ldots,p_L\}$. When $\calC$ is understood, we may write \emph{class $\ell$} as a shorthand for \emph{class $C_\ell$}.
Let $\supp(g)$ be the \emph{support} of a vector $g \in \Fq^k$, i.e., $\supp(g) = \{i \in \{1,\ldots,k\} \colon {g_i \neq 0}\}$. For a packet $x \in \Fq^m$ with coding vector $g \in \Fq^k$, we say that $x$ \emph{belongs} to class $\ell$ if $\supp(g) \subseteq C_\ell$. Let $\lambda(x)$ denote the set of indices of all the classes to which a packet $x$ belongs, i.e., $\lambda(x) \triangleq \{\ell \in \{1,\ldots,L\} \colon \supp(g) \subseteq C_\ell\}$. With a slight abuse of terminology, we will usually refer to a \emph{class $\ell$} to mean all the \emph{data packets} belonging to that class.

Note that, in general, a packet $u_i$ may belong to multiple classes; for instance, we might have $C_1 \cap C_2 = \{i\}$, which implies that $\{1,2\} \subset \lambda(u_i)$. When two classes have non-empty intersection, we will say that these classes \emph{overlap}.

Given a class-based network coding scheme $(\calC,\{p_\ell\})$, every node in the network (including the source node) performs, at each transmission opportunity, the following encoding procedure. First, a class index $\ell$ is randomly selected according to $\{p_\ell\}$. If no packets from that class have yet been received, then the process is repeated until an index $\ell$ is selected such that some packet from class $\ell$ has been received. Then, an outgoing packet is computed as a random linear combination of received packets from class $\ell$.

Let $d_\ell = |C_\ell|$, for $\ell = 1,\ldots,L$. It should be clear that the chunked coding scheme described in Section~\ref{ssec:sparse-disjoint} corresponds to the special case where $\calC$ is a partition of $\{1,\ldots,k\}$, with $d_\ell = d = k/L$, and $\{p_\ell\}$ is uniform. In general, due to the presence of overlapping classes, we may have $\sum_{\ell = 1}^L d_\ell > k$.

Let us now describe the decoding process. For $\ell = 1,\ldots,L$, let $Y^{(\ell)}$ consist of the received packets from class $\ell$, and let $r_\ell = \rank Y^{(\ell)}$. We view $Y^{(\ell)}$, and therefore $r_\ell$, as variables that are constantly updated as new packets are received; in particular, we call the tuple $(r_1,\ldots,r_L)$ the \emph{state} of the receiver. In the context of a decoding process, we say that a class $\ell$ is \emph{decodable} if $r_\ell \geq d_\ell$ and that it \emph{has been decoded} if all the data packets $u_i$ belonging to $\ell$ have been recovered. Decoding starts from some decodable class $\ell$ that has not yet been decoded. This class is decoded by Gaussian elimination. Then, similarly to the decoding of fountain codes, any data packets $u_i$ belonging to $C_\ell$ are \emph{back-substituted} into any overlapping classes, and the ranks $r_1,\ldots,r_\ell$ are recomputed. For instance, if $u_i$ belongs to classes $1$ and $2$, and class $1$ is decoded, then we may imagine that a new packet $y^{(2)}_{n_2+1} = u_i$ has been received. Unless class $2$ has already been decoded, this has the effect of increasing $r_2$ by one unit. The process is then repeated until all classes have been decoded---which is to say that all data packets $u_i$ have been obtained.

The essence of the decoding process is similar to solving a \emph{crossword puzzle}: when a word is ``decoded,'' the recovered letters can be reused to help in the decoding of any overlapping words. Indeed, the idea of a crossword puzzle gives the basis for the simplest nontrivial overlapping scheme, which we call \emph{grid codes}. A simple example of a grid code is given in Fig.~\ref{fig:grid2x2}.
\begin{figure}
\centering
  \subfloat[Chunked code]{
    \includegraphics[scale=1.3]{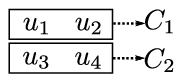}
    \label{fig:chunk2x2}}
    \quad
  \subfloat[Grid code]{
    \includegraphics[scale=1.3]{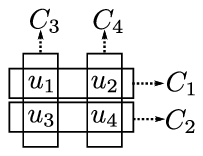}
    \label{fig:grid2x2}}
  \caption{An example of a $2 \times 2$ grid code.}
  \vspace{-2ex}
\end{figure}
A general definition will be presented in Section~\ref{sec:examples}.

\medskip
\begin{example}
  Let $k=4$. The $2 \times 2$ grid code of Fig.~\ref{fig:grid2x2} can be seen as the chunked code $\{C_1,C_2\}$ of Fig.~\ref{fig:chunk2x2} with two extra classes $C_3$ and $C_4$. Let us assume that, in either case, packets from each class are received with equal probability and all received packets are innovative. Suppose that, initially, two packets from $C_1$ have been received, i.e., $r_1 = 2$, so that the decoder is in state $(2,0,0,0)$. For the chunked code to succeed with no overhead, it is necessary that the next two received packets belong to $C_2$, an event that happens with probability $1/4$.

On the other hand, for the grid code to succeed, there is much more flexibility in the possible received packets; more precisely, all the receiver states $(2,2,0,0)$, $(2,1,1,0)$, $(2,1,0,1)$ and $(2,0,1,1)$ are decodable. For instance, suppose that the next two received packets belong to $C_2$ and $C_3$, i.e., the receiver state is $(2,1,1,0)$. Decoding proceeds as follows. First, class 1 is decoded using Gaussian elimination, which yields uncoded packets $u_1$ and $u_2$. Since $u_1$ and $u_2$ are also from classes 3 and 4, respectively, the state is updated to $(2,1,2,1)$. Now class 3 can be decoded, uncovering packet $u_3$. Since $u_3$ also belongs to class 2, the state becomes $(2,2,2,1)$. Now class 2 is decoded, which finally reveals the last packet $u_4$, completing the decoding. Thus, if the initial state is $(2,0,0,0)$ and two more packets are received, the grid code succeeds with probability $4/10 > 1/4$.
\end{example}
\medskip

Let us now examine the issue of decoding complexity. We first describe an alternative way to view the decoding process. Note that, for each new packet $u_i$ that is recovered, one variable is effectively removed from the problem for all the remaining classes. Thus, rather than increasing $r_1,\ldots,r_L$ at each decoding iteration, we can equivalently decrease $d_1,\ldots,d_L$. This has precisely the same effect in the decoding condition $r_\ell \geq d_\ell$. More precisely, let $d_\ell^{(i)}$ denote the size of class $\ell$ (in terms of remaining variables) after the $i$th decoding iteration. Initially, $d_\ell^{(0)} = d_\ell$, for all $\ell$. After the $i$th iteration, when, say, class $\ell^*$ is decoded, we update $d_\ell^{(i)} = d_\ell^{(i-1)} - |C_\ell \cap C_{\ell^*}|$, for all classes that \emph{have not yet been decoded}. We keep $d_{\ell}^{(i)} = d_{\ell}^{(i-1)}$ for the decoded classes, since this tells us precisely the size of the problem that was solved for class $\ell$, i.e., how many packets had to be decoded by Gaussian elimination. Thus, at the end of the decoding process, say, after iteration $t$, we should have $\sum_{\ell=1}^L d_{\ell}^{(t)} = k$, which is precisely the total number of variables. Using this description of the decoding process, we can provide the following bound on the decoding complexity.

\medskip
\begin{theorem}\label{thm:complexity}
 Let $d_{\ell_1},\ldots,d_{\ell_L}$ denote the sizes of all classes sorted in decreasing order. The worst-cast decoding complexity $\chi$, in operations per symbol, is upper bounded by
  \begin{equation}\nonumber
    \chi \leq d_{\ell_t} + \frac{1}{k} \sum_{i=1}^{t-1} d_{\ell_i} (d_{\ell_i} - d_{\ell_t}) \leq d_{\ell_1}
  \end{equation}
where $t$ is the smallest integer such that $\sum_{i=1}^t d_{\ell_i} \geq k$.
\end{theorem}
\begin{proof}
  Without loss of generality, suppose that classes are sorted according to the order in which they are decoded, i.e, class 1 is decoded first, then class 2, and so on. Let $t$ be number of iterations after which decoding is complete. Class 1 is decoded first, after which $d_2 - d_2^{(1)}$ uncoded packets are forwarded to class 2. By examining the matrix of the linear system that has to be solved for class $2$, it is easy to see that this system can be solved with precisely $d_2^{(1)}d_2 m = d_2^{(t)}d_2 m$ operations. In general, each class $\ell$ can be decoded with $d_\ell^{(t)} d_\ell m$ operations, giving a total complexity of
\begin{equation}\nonumber
  \frac{1}{k} \sum_{\ell=1}^L d_\ell^{(t)} d_\ell
\end{equation}
operations per symbol.

To obtain a bound, we need to maximize the function $\sum_{\ell=1}^L d_\ell x_\ell$, subject to the constraints $0 \leq x_\ell \leq d_\ell$, $\ell = 1,\ldots,L$, and $\sum_{\ell=1}^L x_\ell = k$. It is clear that this function is maximized by choosing $x_{\ell_i} = d_{\ell_i}$, $i=1,\ldots,t-1$, and $x_{\ell_t} = k - \sum_{i=1}^{t-1} d_{\ell_i}$, where $t$ is the smallest integer such that $\sum_{i=1}^t d_{\ell_i} \geq k$. Thus, we obtain
\begin{align}
\chi &\leq \frac{1}{k} \sum_{i=1}^{t-1} d_{\ell_i}^2 + \frac{1}{k} \left(k - \sum_{i=1}^{t-1} d_{\ell_i}\right) d_{\ell_t} \nonumber \\
&= d_{\ell_t} + \frac{1}{k} \sum_{i=1}^{t-1} d_{\ell_i} (d_{\ell_i} - d_{\ell_t}) \nonumber
\end{align}
with equality if $\ell_i = i$, $i=1,\ldots,t$, and $C_1,\ldots,C_t$ are disjoint.

The second expression follows from $d_{\ell_i} \leq d_{\ell_1}$ and $\sum_{i=1}^{t-1} d_{\ell_i} \leq k$. We have
\begin{equation}\nonumber
  \chi \leq d_{\ell_t} + \frac{1}{k} \sum_{i=1}^{t-1} d_{\ell_i} (d_{\ell_1} - d_{\ell_t}) \leq d_{\ell_1}.
\end{equation}
\end{proof}
\medskip

Theorem~\ref{thm:complexity} shows that the complexity is dominated by the largest $t$ classes and is \emph{not} increased by adding \emph{any} number of classes that are smaller than the largest $t$ classes. In particular, for a code with fixed-size classes, the complexity is never greater than that of the corresponding chunked code. For general codes, we should in fact expect a complexity much smaller than the bound of Theorem~\ref{thm:complexity}. This is because that bound is achieved when the first $t$ classes to be decoded are the largest ones and are disjoint, while in practice we would expect smaller classes to be decoded first and be back-substituted into larger ones.

%for disjoint classes; and smaller classes are usually more likely to be decoded first.
%In fact, we should expect a smaller complexity in general, since the bound in Theorem~\ref{thm:complexity} is achieved for disjoint classes.

%A simple upper bound is given by the complexity of performing Gaussian elimination on all classes, which is $O(\sum_{\ell=1}^L d_\ell^2 m)$. Let $K = \sum_{\ell=1}^L d_\ell$ and let $\alpha = K/k$. The parameter $\alpha$ gives the average number of classes to which a packet belongs; we call it the \emph{covering factor} of the code. In the special case that $d_\ell = d$ for all $\ell$, we have $\sum_{\ell=1}^L d_\ell^2 m = L d^2 m =  \alpha d km$. Thus, we see that, in this case, the complexity is increased by at most $\alpha$ times as compared with a chunked code. In practice, however, we should expect a much smaller complexity, since the back-substitution process essentially bypasses several steps of Gaussian elimination.

Evaluating the performance is a much harder issue. This is due to the fact that Gaussian elimination is combined with back substitution in a recurring manner, leading to an extremely intricate decoding process. Nevertheless, for simple cases, we can compute the performance exactly. Fig~\ref{fig:2x2} shows the exact probability of successful decoding versus overhead for the $2 \times 2$ grid code of Fig.~\ref{fig:grid2x2}.
\begin{figure}
\centering
\includegraphics[scale=0.9]{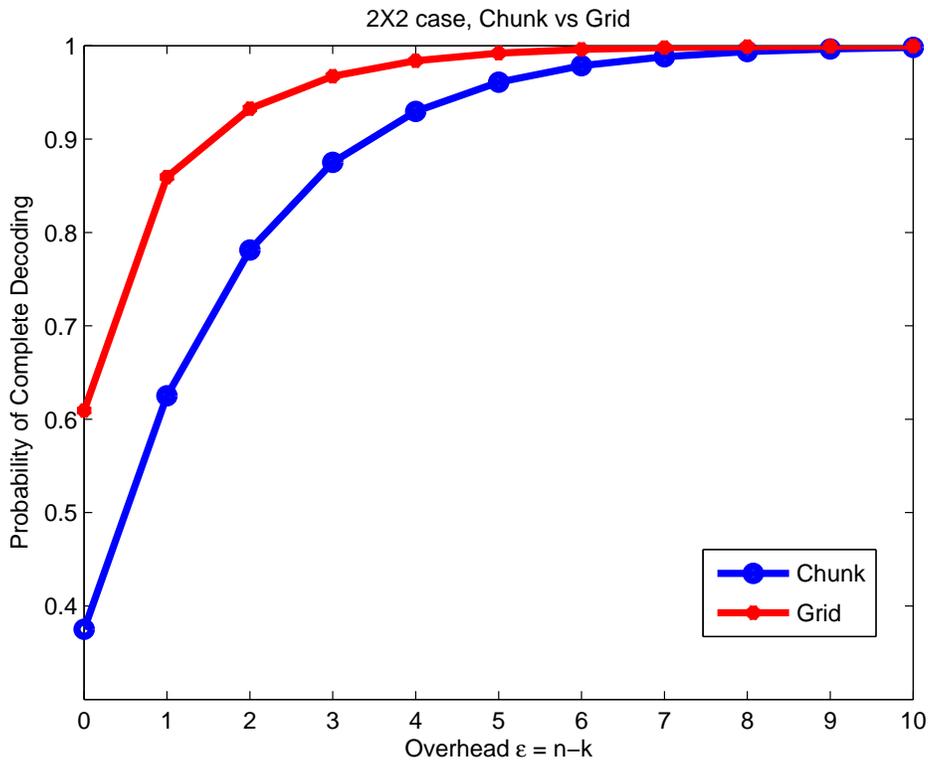}
%\vspace{-5mm}
\caption{Comparison between a $2 \times 2$ grid code and a chunked code with $d=2$ and $k=4$.}
\label{fig:2x2}
\vspace{-2ex}
\end{figure}
It can be seen that, for the same complexity, the performance of this grid code \emph{uniformly better} than that of the corresponding chunked code.

\section{Examples of Codes}
\label{sec:examples}

In this section, we present some examples of codes with overlapping classes. The performance of these codes will be investigated in Section~\ref{sec:performance}.

\medskip
\begin{definition}\label{def:rectangular-grid}
  Let $k = d d'$. A \emph{$d' \times d$ (rectangular) grid code} $\calC = \{C_1,\ldots,C_L\}$ consists of $L = d + d'$ classes given by
 \begin{align}
 C_i &= \{(i-1) d + j \mid j=1,\ldots,d\},\quad i=1,\ldots,d' \nonumber \\
 C_{d' + j} &= \{(i-1) d + j \mid i=1,\ldots,d'\},\quad j=1,\ldots,d. \nonumber
 \end{align}
\end{definition}
\medskip

Note that, in a $d' \times d$ grid code, the first $d'$ classes have size $d$ and form a partition of $\{1,\ldots,k\}$, while the last $d$ classes have size $d'$ and also form a partition of $\{1,\ldots,k\}$.

When all classes have the same size, i.e., $d' = d$, we obtain
\emph{square} grid codes. These codes are, unfortunately, too restrictive, since we must have $k = d^2$. A way to span a higher number of packets with fixed-size classes is provided by \emph{diagonal} grid codes. For convenience, in the next definitions, assume that packet and class indices are numbered starting at zero.

\medskip
\begin{definition}\label{def:diagonal-grid}
  Let $k = L_0 d$ and assume $L \leq (L_0)^2$. A \emph{$(k,d,L)$ diagonal grid code with angle set $\Theta = \{\theta_0,\ldots,\theta_{L_0-1}\}$} consists of $L$ classes given by
\begin{equation}\nonumber
  C_\ell = \left\{(i+j\theta_s) d + j \bmod k \mid j=0,\ldots,d-1 \right\},\quad%\\
  i=\ell \bmod L_0,\; s = \lfloor \ell/L_0 \rfloor,\; \ell = 0,\ldots,L-1.
\end{equation}
A \emph{diagonal grid code with angle $\theta$} is a diagonal grid code with angle set $\Theta = \{0,\theta,2\theta,\ldots\}$.
\end{definition}
\medskip

For $s = 0,\ldots,\lceil L/L_0 \rceil - 1$, the classes $C_{sL_0}$,$\ldots$,$C_{(s+1)L_0-1}$  form a partition of $\{0,\ldots,k-1\}$. In particular, the $L_0$ classes with angle $\theta_s = 0$ correspond to a chunked code. An example of a diagonal grid code is given in Fig.~\ref{fig:diagonal-grid}.
\begin{figure}
\centering
\includegraphics[scale=1.3]{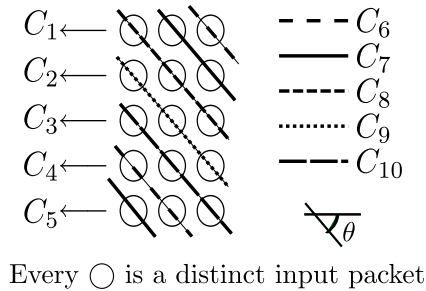}
%\vspace{-5mm}
\caption{Example of a $(15,3,10)$ diagonal grid code.}
\label{fig:diagonal-grid}
\vspace{-2ex}
\end{figure}

The design of a diagonal grid code minimizes the maximum size of the intersection of two classes. It is easy to see that, if all the nonzero $\theta_s$ are relatively prime to $L_0$, then any two distinct classes overlap in at least $\lfloor L/L_0 \rfloor - 1$ and at most $\lceil L/L_0 \rceil - 1$ classes. To see that this value is optimal, consider a bipartite graph with packets and classes as nodes, and an edge connecting a packet to a class if that packet belongs to that class. It follows that the maximum degree of a packet must be at least the average degree $(Ld)/(L_0 d) = L/L_0$.

Note that a diagonal grid code consists essentially of multiple layers of chunked codes each applied after the packets $0,\ldots,k-1$ undergo a certain (grid-like) permutation. Thus, the construction of Definition~\ref{def:diagonal-grid} can be generalized by using arbitrary permutations. For $s = 0,\ldots,\lceil L/L_0 \rceil - 1$, let $\pi_s$ be a permutation of $\{0,\ldots,k-1\}$. Then we may consider a code with $L$ classes of size $d$ given by
\begin{equation}\nonumber
  C_\ell = \left\{\pi_s(i d + j) \mid j=0,\ldots,d-1 \right\},\quad%\\
  i=\ell \bmod L_0,\; s = \lfloor \ell/L_0 \rfloor,\; \ell = 0,\ldots,L-1
\end{equation}
where $L_0 = k/d$. Without loss of generality, we will assume that $\pi_0$ is the identity permutation. If all the remaining permutations are chosen uniformly at random, we will call the resulting code a \emph{random-layer code}.

For generality, in all the codes described above, we have left the probability distribution $\{p_\ell\}$ unspecified. However, in the case that all classes have a constant size $d$, it is quite natural to use a uniform distribution $p_\ell = 1/L$ for all $\ell$. More generally, we see no reason to assign different probabilities for classes of the same size, and we will use this assumption in all the experiments in the next section. %[[and our experiments seem to confirm this hypothesis??]]

\section{Performance Evaluation}
\label{sec:performance}

In this section, we use simulations to evaluate the performance of the codes described in the previous section.

We make the following assumptions:
\begin{enumerate}
  \item All received packets are linearly independent whenever possible, i.e., $\rank A^{(\ell)}B^{(\ell)} = \min\{n_\ell, d_\ell\}$, for all $\ell$.
%  \item For all $\ell$, $P[\text{$y_j$ belongs to $C_\ell$}] = p_\ell$.
  \item The probability that a received packet belongs to class~$\ell$ is exactly equal to $p_\ell$, for all $\ell$.
\end{enumerate}

Note the two assumptions above concern themselves with the network topology and the network code, and they are required if we wish to pursue an analysis that is independent of the network. Assumption 1 implies that the source node must generate a sufficient number of packets from each class ($N_\ell \geq d_\ell$) and that both the encoding at the source node and the network code must not introduce any linear dependence on any set of up to $d_\ell$ received packets. Assumption 2 means that the network preserves the designed probability distribution on classes. Both assumptions should hold true if $q$ and each $d_\ell$ are sufficiently large. In order to satisfy this requirement, we assume that a parameter $d_{\min}$ is given such that any valid code must satisfy $d_\ell \geq d_{\min}$, for all $\ell$. Specifically, we consider $d_{\min} = 25$ in the following results. Note that the value of $q$ does not affect code design.

Performance is evaluated in terms of the complexity-overhead tradeoff. Since the problem is inherently delay-tolerant---each receiver is interested in receiving the complete file with probability 1, no matter how long it takes---the two main figures of merit are the expected complexity and the expected overhead. Note that the figure of expected overhead automatically incorporates the probability of failure for each specific overhead, therefore eliminating the need to consider a three-dimensional tradeoff space.

Fig.~\ref{fig:plot1} shows how complexity is traded off against overhead in a chunked coded. %
\begin{figure}
\centering
\includegraphics[scale=0.9]{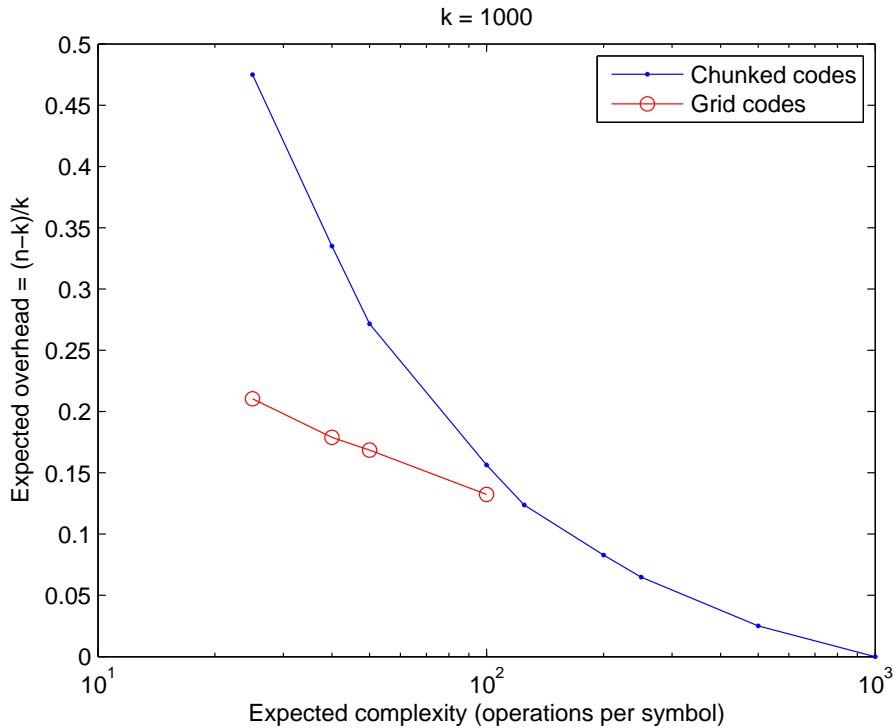}
%\vspace{-5mm}
\caption{Performance of chunked codes and diagonal grid codes for $k=1000$.}
\label{fig:plot1}
\vspace{-2ex}
\end{figure}
At one extreme, we have a dense code with a single class of size $d=k$; this code has optimal overhead but prohibitively large complexity. At the other extreme we have a chunked code with class size $d = d_{\min}$, which attains the minimum possible complexity at the expense of a large overhead. As shown in Fig.~\ref{fig:plot1}, for small to moderate complexity, diagonal grid codes can outperform chunked codes by a large margin. Note that the complexity of diagonal grid codes is precisely equal to the class size $d$. The number of classes $L$ for each grid code has been tuned experimentally to maximize the performance for the given parameters. From left to right, the points in Fig.~\ref{fig:plot1} correspond to $L=28$, $12$, $9$, $2$.

Fig~\ref{fig:plot2} shows similar results for a scenario where $k=4096$. As one can see, well-designed grid codes significantly outperform chunked codes. From left to right, the grid codes in the figure have $L=207$, $92$, $43$. Fig~\ref{fig:plot2} also shows the performance of codes with varying class sizes, referred to as mixed codes.
\begin{figure}
\centering
\includegraphics[scale=0.9]{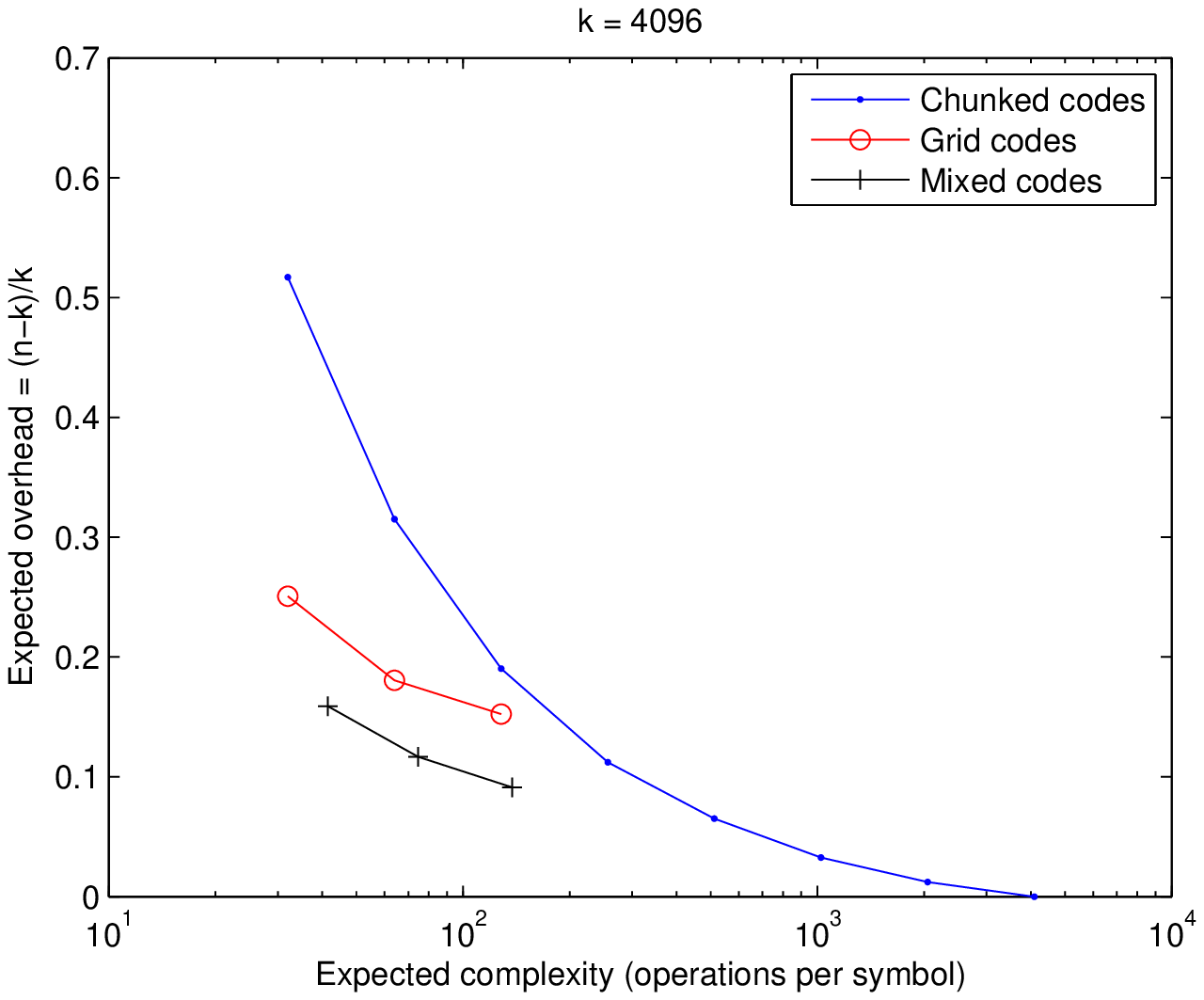}
%\vspace{-5mm}
\caption{Performance of chunked codes, diagonal grid codes and mixed codes for $k=4096$.}
\label{fig:plot2}
\vspace{-2ex}
\end{figure}
From left to right, these codes are: a $(4096,32,200)$ diagonal grid code with an additional random class of size 2048; a $(4096,64,86)$ diagonal grid code with an additional random class of size 1024; and a $(4096,128,38)$ diagonal grid code with an additional random class of size 512. In all cases, the distribution $\{p_\ell\}$ used is the uniform one. In comparison with their corresponding grid codes, the mixed codes exhibit a significantly lower overhead with only a marginal increase in complexity.
As discussed in Section~\ref{sec:sparse-overlapping}, this is due to the fact that the extra (large) class is typically decoded only after many other (smaller) classes have been decoded and back-substituted. The effect of a large class is analogous to that of a high degree check in LT codes: establishing a ``bridge'' between non-overlapping classes and thus allowing the decoding ``ripple'' \cite{Shokrollahi2006} to be maintained for a longer time.

Our results show that, for a fixed expected complexity, the use of overlapping classes can reduce the expected overhead by up to 70\%.

\section{Concluding Remarks}
\label{sec:conclusion}

This paper presents a novel approach to network coding based on the concept of overlapping classes. The approach generalizes chunked coding and allows a propagative decoder that enjoys many of the benefits of fountain codes. Our proposed scheme, while still in its initial stages, seems to be a promising step towards a full network coding solution to peer-to-peer file distribution. More generally, our approach seems to be suitable for any application that would benefit from a combination of fountain coding and network coding.

We remark that, while our analysis assumes no feedback between nodes, nothing prevents a protocol based on our scheme from using control messages to aid the communication. By carefully designing the amount of protocol overhead, the overall performance of the scheme may be further increased.

Our main objective with this paper has been to suggest a new possible direction for research in network coding, as more questions remain open than have been answered here (especially in the theoretical side). In particular, the design of good codes with constant or non-constant class sizes (and possibly nonuniform distribution) is an important open problem. Due to the recursive nature of the decoding process, the development of analytical bounds on performance also remains elusive at this point. We hope to address both problems in our future work.

\bibliographystyle{IEEEtran}
\bibliography{IEEEabrv,networkcoding,codingtheory,rankmetric,silva}

\end{document}